\documentclass[10pt,conference]{IEEEtran}
\IEEEoverridecommandlockouts
\usepackage{cite}
\usepackage{url}
\usepackage{amsmath,amssymb,amsfonts,amsthm}
\usepackage{mathtools}
\usepackage{booktabs}
\usepackage{algorithm,algpseudocode}
\usepackage{graphicx}
\usepackage{svg}
\usepackage{textcomp}
\usepackage{xcolor}
\usepackage{hyperref}
\usepackage{physics}
\usepackage{bm}
\usepackage{array}
\usepackage[normalem]{ulem} 
\usepackage{tikz}
\usetikzlibrary{decorations.pathreplacing}

\DeclarePairedDelimiterX{\Expval}[2]{\langle}{\rangle}{#2\,\delimsize\vert\,#1\,\delimsize\vert\,#2}

\newtheorem{lemma}{Lemma}
\newtheorem{theorem}{Theorem}

\newcommand{\BigO}{\mathcal{O}}
\newcommand{\WF}{\psi(\bm{\theta})}
\newcommand{\red}[1]{\textcolor{red}{#1}}
\newcommand{\bfred}[1]{\textbf{\red{#1}}}

\usepackage{subcaption}
\newcommand{\sidecaption}[1]
{\raisebox{\abovecaptionskip}{\begin{subfigure}[t]{1.6em}
			\caption[singlelinecheck=off]{}
			\label{#1}
	\end{subfigure}}\ignorespaces}

\def\BibTeX{{\rm B\kern-.05em{\sc i\kern-.025em b}\kern-.08em
    T\kern-.1667em\lower.7ex\hbox{E}\kern-.125emX}}

\title{Implementing Slack-Free Custom Penalty Function for QUBO on Gate-Based Quantum Computers}

\begin{document}
\author{\IEEEauthorblockN{Xin Wei LEE}
\IEEEauthorblockA{\textit{School of Computing and Information Systems} \\
\textit{Singapore Management University}\\
80 Stamford Rd, Singapore \\
xwlee@smu.edu.sg}
\and
\IEEEauthorblockN{Hoong Chuin LAU}
\IEEEauthorblockA{\textit{School of Computing and Information Systems} \\
\textit{Singapore Management University}\\
80 Stamford Rd, Singapore \\
hclau@smu.edu.sg}}

\maketitle

\begin{abstract}
    Solving NP-hard constrained combinatorial optimization problems using quantum algorithms remains a challenging yet promising avenue toward quantum advantage.
    Variational Quantum Algorithms (VQAs), such as the Variational Quantum Eigensolver (VQE), typically require constrained problems to be reformulated as unconstrained ones using penalty methods.
    A common approach introduces slack variables and quadratic penalties in the QUBO formulation to handle inequality constraints.
    However, this leads to increased qubit requirements and often distorts the optimization landscape, making it harder to find high-quality feasible solutions.
    To address these issues, we explore a slack-free formulation that directly encodes inequality constraints using custom penalty functions, specifically the exponential function and the Heaviside step function.
    These step-like penalties suppress infeasible solutions without introducing additional qubits or requiring finely tuned weights.
    Inspired by recent developments in quantum annealing and threshold-based constraint handling in gate-based algorithms, we implement and evaluate our approach on the Multiple Knapsack Problem (MKP).
    Experimental results show that the step-based formulation significantly improves feasibility and optimality rates compared to unbalanced penalization, while reducing overall qubit overhead.
\end{abstract}

\section{Introduction}
Solving NP-hard constrained combinatorial optimization problems on quantum computers poses a significant challenge due to the need to enforce complicated constraints within the quantum framework.
Such problems are not only theoretically intractable at scale but also ubiquitous in domains like logistics and finance, motivating intense interest in quantum algorithms that might surpass classical methods.
Quantum annealers and gate-based variational quantum algorithms (VQA), notably the Quantum Approximate Optimization Algorithm (QAOA)~\cite{farhi2014quantum} and the Variational Quantum Eigensolver (VQE)~\cite{vqe}
have emerged as promising approaches for tackling these NP-hard problems.
These algorithms work by encoding the optimization problem into a Hamiltonian whose ground state encodes an optimal (or near-optimal) solution. However, incorporating hard constraints
into such quantum formulations is non-trivial, as the naive approach of directly mapping a constrained problem often requires additional encoding techniques.

In practice, the standard strategy is to convert a constrained problem into an equivalent unconstrained one, known as quadratic unconstrained binary optimization (QUBO),
by adding penalty terms to the objective function that energetically penalize any violation of the constraints.
Inequality constraints are typically handled by introducing extra slack binary variables that transform them into equalities before applying a squared penalty.
Although this penalty method ensures that any constraint violation increases the cost, it comes with serious drawbacks.
In particular, the squared-error penalties introduce dense higher-order interactions among the original decision variables.
In quantum annealing implementations, such dense coupling exceeds the sparse hardware connectivity, necessitating complex embeddings that use multiple physical qubits per logical variable.
In gate-based variational algorithms, the added slack qubits and coupling terms dramatically expand the state search space and create a rugged energy landscape,
making it harder for the hybrid quantum-classical optimizer to converge~\cite{in-constraint,assist-cvrp}.
This can lead to suboptimal solutions or infeasible outputs if the penalty weight is not perfectly tuned, underscoring the difficulty of faithfully enforcing constraints using conventional quadratic penalty encoding.

An alternative approach to handling constraints was pioneered by Ohzeki et al. in the context of quantum annealing~\cite{ohzeki-inequality,ohzeki-subgradient}.
Instead of relying on heavy penalty terms, they formulated the partition function of the problem so that only states satisfying the constraints contribute.
In this method, each constraint is enforced by a Heaviside step function included in the partition function, which evaluates to 1 if the constraint holds (and 0 if it is violated).
This effectively filters out infeasible configurations without requiring a large penalty coefficient.
As a result, the Ohzeki method can impose inequality constraints via linear couplings to auxiliary Lagrange-multiplier-like variables, rather than via dense all-to-all penalty interactions.
Empirical studies have shown that enforcing constraints in this manner allows larger and more densely connected problems to be solved on current quantum annealers than the traditional penalty mapping would allow.
However, the step function (or other penalty function) cannot be incorporated directly in quantum annealers, resorting to solving a Lagrangian relaxation in the Ohzeki method.
This still requires extra effort to be solved, such as using the subgradient method to update the Lagrangian multipliers.

Encouraged by such developments, recent gate-based quantum algorithms have started to explore slack-free methods that more faithfully encode constraints while avoiding slack-variable
overhead~\cite{fischer-lagrangian,subgradient-1qbit,vqec,lagrangian-qubo}.
Golden et al. replace the standard phase separator in QAOA with a threshold function that outputs 1 only for solutions above a certain objective value threshold~\cite{th-gm-qaoa}.
This effectively acts as a step-function filter, focusing the algorithm on feasible solutions.
On the other hand, Monta\~nez-Barrera et al. propose an unbalanced penalization technique that imposes a much larger penalty for violating an inequality constraint than
for satisfying it~\cite{ub-dwave,ub-penalty}.
This asymmetric penalty mimics a step function by heavily suppressing infeasible configurations without requiring extra slack qubits.
However, these methods also come with drawbacks. The threshold-based QAOA requires extra ancilla qubits to implement a quantum comparator, which will trigger the flag qubit if a constraint is violated.
The Grover mixer can also be expensive in terms of the number of gates in the circuit, as it requires an initialization of a superposition of all feasible states~\cite{grover-mixer,constraint-preserving-qaoa}.
For unbalanced penalization, the quadratic approximation of the exponential penalty can lead to false minima in the optimization landscape, i.e., the minimum in the unbalanced QUBO is not the minimum
in the original constrained problem.


Building on these insights, this work introduces a slack-free VQE-based solver that leverages a custom penalty function to enforce the constraints.
We use two examples of the penalty function: the exponential function and the Heaviside step function to solve the Multiple Knapsack problem (MKP).
In contrast of the previous methods, the penalty function acts on the constraint Hamiltonian itself and is evaluated directly using classical computers, breaking the restriction of keeping a QUBO formulation
in quantum annealers, as well as not requiring extra qubits in the threshold-based algorithms.
This formulation preserves the feasibility of low-energy states by construction without the need for fine-tuned penalty weights.
Although requiring a complexity that scales exponentially with the number of terms in the constraint, 
the experimental results demonstrate that the proposed method attains a higher feasibility rate and optimality rate compared to the unbalanced penalization method.
Overall, the main contribution of this paper is to show that variational quantum algorithms can incorporate custom penalty functions, e.g., step function, to solve inequality-constrained problems like MKP
with increased feasibility and accuracy, and with fewer qubits, marking a promising advance in quantum approaches to constrained combinatorial optimization.

\section{Background}
\subsection{QUBO formulation of constrained optimization}
A constrained binary program is generally defined as
\begin{align}
    \min_x\ & f(x) \\ 
    \text{s.t.}\ & g_i(x) = 0,\quad i = 1,...,m, \\
    & h_j(x) \leq 0,\quad j= 1,...,k,
\end{align}
where $x\in\{0,1\}^n$ is an $n$-bit binary string that represents the variable of the problem. The aim of a solver is to find the solution bit string $x$ that minimizes
the objective $f(x)$ while keeping all the constraints $g_i(x)$ and $h_j(x)$ satisfied. Any solution $x$ that satisfies all the constraints is called
a feasible solution, and a solution that violates any of the constraints is called an infeasible solution. 

When solving an unconstrained problem (a problem that only has $f(x)$ without constraints) using variational quantum algorithms (VQA),
the objective is converted to the Ising model, i.e., the binary variables $x_i\in\{0,1\}$ are converted to the spin variables $z_i\in\{-1,1\}$
using a linear map $x_i=(1-z_i)/2$. Then, the spin variables are replaced by the Pauli-Z matrix to obtain the problem Hamiltonian $H$.
Therefore, the conversion can be simplified as
\begin{equation}
    x_i \rightarrow \frac{1-Z_i}{2}.
    \label{eqn:binary-to-spin}
\end{equation}
After the problem Hamiltonian is formulated, the expectation
$\expval{H}{\WF}$ is minimized with respect to some parameterized trial state $\ket{\WF}$ to obtain the solution to the original problem, where $\bm{\theta}=(\theta_1,\theta_2,...,\theta_n)$
is a collection of the variational parameters.
Since constraints cannot be encoded directly into quantum computers, a common way is to formulate it as an unconstrained problem by adding penalty terms into
the objective so that they penalize the infeasible solution. In most application cases, the objective $f(x)$ is either linear or quadratic, and the constraints
$g_i(x)$ and $h_j(x)$ are linear, so the resulting expression is a quadratic function of $x$. This is known as the
quadratic unconstrained binary optimization (QUBO) formulation. In QUBO, the equality constraints are formulated as $\lambda g(x)^2$, where
$\lambda$ is the penalty coefficient to decide how much to penalize. This is done so that if $g(x) \neq 0$, this term will return a large value that will increase
the objective that is to be minimized. The inequality constraints can also be formulated with the same reasoning. An inequality is converted to an equality with the
introduction of slack variables, i.e., $h(x)+y=0$, yielding the penalty term $\mu(h(x) + \sum_{k=0}^{N-1}2^ky_k)^2$,
where $y\geq 0$ is the slack variable in its binary representation $2^ky_k$. Hence, the resulting QUBO loss function becomes
\begin{multline}
    \mathcal{L}_{\text{slack}}(x) = f(x) + \sum_{i=1}^m \lambda_0^{(i)} g_i(x)^2 \\
    + \sum_{j=1}^k \lambda_1^{(j)}\left(h_j(x) + \sum_{l=0}^{N-1}2^ly_l \right)^2.
    \label{eqn:loss-slack}
\end{multline}
The slack variable formulation is not a good method when optimizing with VQAs as it introduces extra variables, meaning that extra qubits, which is a very
limited resource in current quantum devices, are needed to map the original problem to a QUBO. Moreover, many have reported that the optimization landscape 
of the slack QUBO formulation is not faithful to the objective of the original problem, i.e., lower loss in QUBO does not necessarily mean better objective in the
original problem due to different possible values that a slack variable can hold to satisfy the equality. 
Therefore, it calls for the slack-free QUBO formulation to overcome the above drawbacks in variational quantum optimization.

\begin{figure}
    \centering
    \includegraphics[width=0.95\linewidth]{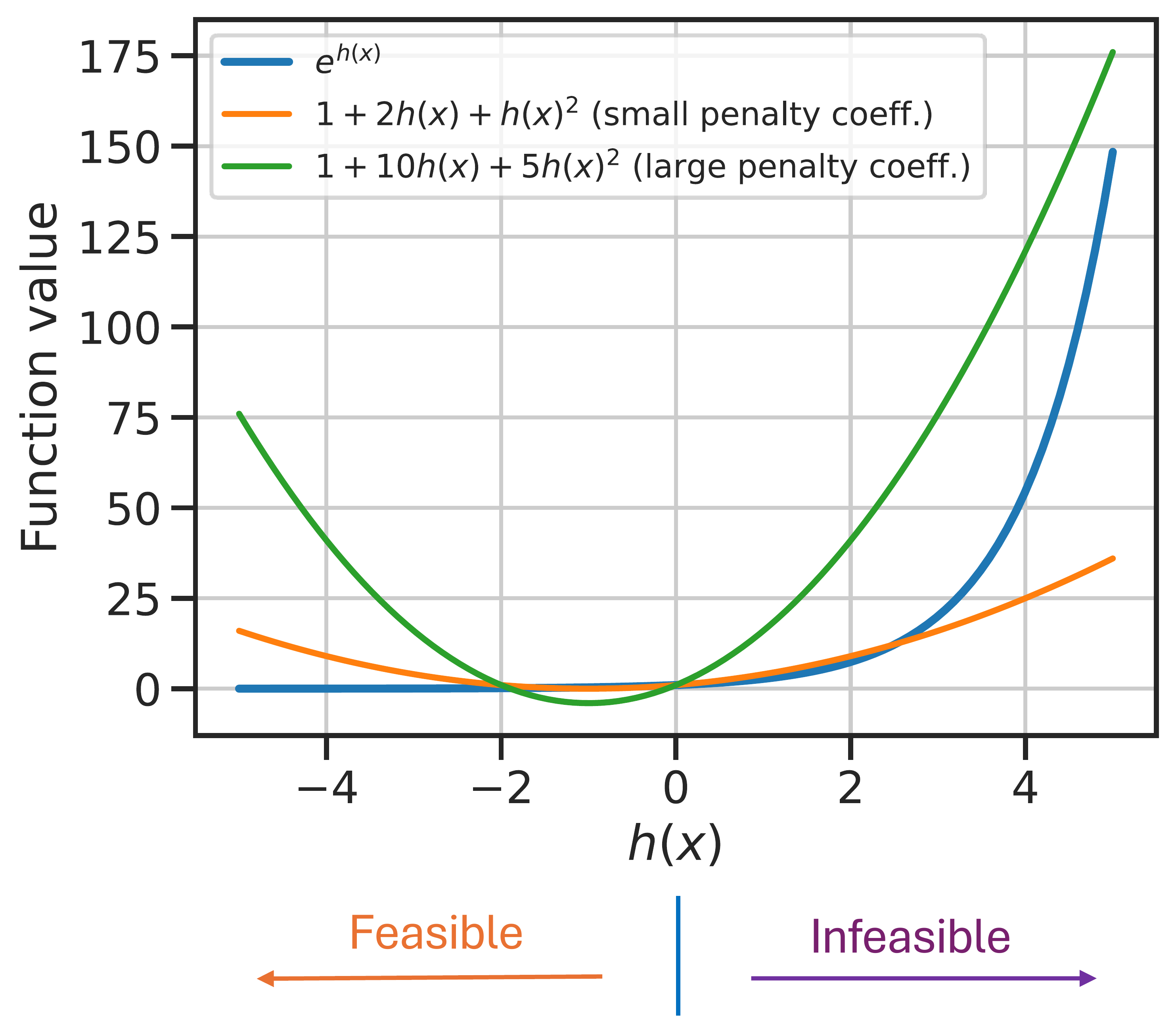}
    \caption{The effect of different penalty function. The penalty function is supposed to penalize (output large values) for infeasible solution $h(x)>0$ and output small values %
    for feasible solution $h(x)\leq 0$.}
    \label{fig:functions}
\end{figure}

There are several methods to formulate the constrained problem without slack variables, including Lagrangian relaxation~\cite{vqec}
and unbalanced penalization~\cite{ub-dwave,ub-penalty}. We focus on the unbalanced penalization method.
Unbalanced penalization eliminates the use of slack variables by replacing the inequality penalty term with an exponential function of the inequality
$e^{h(x)}$, so that the penalty term is small when $h(x)\leq 0$ (constraint satisfied) and exponentially large when $h(x)>0$ (constraint violated). 
However, adding exponential term will cause the loss function to be no longer quadratic, ruining the structure of QUBO. This causes difficulty in implementing
the exponential loss function in quantum annealers and also quantum observables in quantum circuits. Hence, the exponential term is approximated up to the second
order in~\cite{ub-dwave,ub-penalty} (using the Taylor expansion) to maintain the quadratic loss function: $e^{h(x)} \approx 1+h(x)+h(x)^2/2$.
The resulting QUBO loss function for the unbalanced penalization is
\begin{multline}
    \mathcal{L}_{\text{unb}}(x) = f(x) + \sum_{i=1}^m \lambda_0^{(i)} g_i(x)^2 \\
    + \sum_{j=1}^k \left[\lambda_1^{(j)} h_j(x) + \lambda_2^{(j)}h_j(x)^2 \right],
    \label{eqn:loss-ub}
\end{multline}
where $\lambda_{0,1,2}$ are the penalty coefficients assigned to tune the strength of each penalty term. 
If we assume a linear or quadratic $f(x)$, and linear $g(x)$ and $h(x)$, we can see that $\mathcal{L}_{\text{unb}}(x)$ is quadratic,
so it is still a QUBO and can be solved using a quantum annealer. However, the quadratic approximation of the exponential function has
its own caveats. When the magnitude of $h(x)$ becomes large, although the behavior of the quadratic curve works as desired (gives a large penalty)
for $h(x)>0$, it is not desirable for $h(x)\leq 0$ as the penalty also becomes large when it is supposed to be small
(as in the exponential curve, see Fig.~\ref{fig:functions}). Consequently, when the penalty coefficients are not tuned properly,
the loss function will have a false minimum, i.e., the minimum of $\mathcal{L}_{\text{unb}}(x)$ is not equal to the minimum of the original problem.
Therefore, we attempt to improve the formulation by considering the direct computation of the penalty function.

\subsection{Variational quantum eigensolver}
Variational quantum eigensolver (VQE) is a common method used in quantum computing to find the minimum energy of a given Hamiltonian.
The energy in this case corresponds to the objective of the QUBO loss function. The objective of each solution bit-string is encoded in the so-called
problem Hamiltonian. The problem Hamiltonian $H$ is constructed using the conversion mapping Eq.~(\ref{eqn:binary-to-spin}).
The minimum energy (eigenvalue) of $H$ is then searched using a parameterized trial state $\ket{\WF}$ (also known as an ansatz), where $\bm{\theta}$ is a list of
variational parameters to be trained.
Minimizing the expectation of $H$ with respect to $\ket{\WF}$ is equivalent to minimizing the QUBO loss function:
\begin{equation}
    \min_{\bm{\theta}} \expval{H}{\WF} = \min_x\mathcal{L}(x).
\end{equation}
The optimal $\bm{\theta}$ is usually found using classical optimizers. The trial state can be generated with, e.g., hardware-efficent ansatz (HEA)~\cite{hea,practical-hea},
which consists of alternating layers of rotational gates and entangling gates.

\section{Implementation of penalty function}
For simplicity, consider a binary minimization program with only linear inequality constraints:
\begin{align}
    \min_x\ & f(x) \\
    \text{s.t.}\ & Ax\leq b,
\end{align}
where $x$ and $b$ are vectors of length $n$ (number of variables) and $A$ is a $k\times n$ matrix containing the 
coefficients of $k$ constraints. Let $h_i(x)$ denote the $i$-th constraint:
\begin{equation}
    h_i(x) = \sum_{j=1}^n a_{ij}x_j - b_i \leq 0.
    \label{eqn:hix}
\end{equation}
We want to implement a penalty function on the inequality constraints $\xi:h_i(x)\rightarrow \mathbb{R}$,
such that the value of $\xi$ is large when the solution is infeasible $(h_i(x)>0)$ and small when the solution is feasible $(h_i(x)\leq 0)$.
A good choice of penalty function should separate the values returned by feasible solutions and infeasible solutions distinctively. Potential candidates
of the penalty function are the exponential function as suggested in~\cite{ub-penalty}, hyperbolic tangent function, the sigmoid function,
or the Heaviside step function. However, these functions are usually not polynomial in terms of $x$, so they are difficult to be implemented on quantum devices.

The loss function for the formulation with custom penalty function is 
\begin{equation}
    \mathcal{L}_{\text{custom}}(x) = f(x) + \sum_{i=1}^k \lambda_i\ \xi[h_i(x)].
    \label{eqn:loss-custom}
\end{equation}
Applying the conversion in Eq.~(\ref{eqn:binary-to-spin}), we have
\begin{equation}
    H = H_f + \sum_{i=1}^k \lambda_i\ {\xi}(H_i).
    \label{eqn:custom-ham}
\end{equation}
$H_f$ and $H_i$ are obtained by applying the linear map to $f(x)$ and $h_i(x)$ respectively.
Note that $h_i(x)$ are linear, so $H_i$'s are expressed in terms of only the sums of one-local Pauli-Z strings (only one $Z$ in each Pauli string).
Since $H_i$ is diagonal, the function ${\xi}$ acts directly on each diagonal element of $H_i$:
\begin{equation}
    \xi(H_i) =
    \begin{pmatrix}
    \xi(\mu_1) &&& \\ 
    & \xi(\mu_2) && \\ 
    && \ddots & \\ 
    &&& \xi(\mu_N)
    \end{pmatrix},
\end{equation}
where $\mu_j$ is the $j$-th eigenvalue of $H_i$ and $N=2^n$ ($n$ is the number of qubits in the system).

The goal of variational quantum optimization is to minimize the expectation of the problem Hamiltonian $H$, hence we need to compute
\begin{equation}
    \expval{H} = \expval{H_f} + \sum_{i=1}^k \lambda_i\expval{\xi(H_i)}.
\end{equation}
$\expval{H}$ can be computed using measurements from the Pauli-Z observable, but $\expval{\xi(H_i)}$ cannot be computed in such a way
as it is a function of the sum of Pauli-Z strings. Hence, we need to compute $\expval{\xi(H_i)}$ separately using classical computers.
In most cases, the expectation is calculated as
\begin{equation}
    \expval{\xi(H_i)} = \sum_{x=0}^{N-1} \xi(\mu_x)p_x.
    \label{eqn:exp-penalty}
\end{equation}
$p_x$ is probability of sampling the basis state $\ket{x}$ (corresponding to the solution bit-string $x$) with eigenvalue $\xi(\mu_x)$.
Note that $x$ can also represent the decimal form of the solution bit-string, as we do not distinguish the decimal and binary forms here.
For shot-based quantum simulation,
$p_x$ can be obtained by sampling the quantum circuit. Two quantities are required to compute the expectation: $\mu_x$ (so that we can compute $\xi(\mu_x)$)
and $p_x$. The probabilities $p_x$ can be obtained by sampling quantum circuits, so what is left is the eigenvalue $\mu_x$.

Since most of the problems in applications only has linear constraints, we consider only the constraint Hamiltonian as the sum of one-local Pauli-Z's.
Applying the linear transformation to convert the binary variables to Pauli-Z on Eq.~(\ref{eqn:hix}), $H_i$ can be explicitly written as
\begin{align}
    H_i & = \sum_{j=1}^n a_{ij}\left(\frac{1-Z_j}{2}\right) - b_i \\
    & = C_i - \frac{1}{2} \tilde{H}_i. \label{eqn:ci-hi}
\end{align}
Eq.~(\ref{eqn:ci-hi}) consists of two terms: $C_i = \tilde{A}_i - b_i$, where $\tilde{A}_i = \sum_{j=1}^n a_{ij}/2$ is half of the sum of all coefficients in constraint $i$ and; $\tilde{H}_i = \sum_{j=1}^n a_{ij}Z_j$ encodes the eigenvalues of different combination of bit-strings (the eigenstates).
Here, we let $\tilde{H}$ be a general expression in the form of a sum of one-local Pauli-Z:
\begin{equation}
    \tilde{H} = \sum_{j=1}^n c_jZ_j
    \label{eqn:weighted-sum}
\end{equation}
By using some of the properties of the Pauli-Z matrix, the eigenvalues of $\tilde{H}$ can be inferred.
\begin{lemma}
    Given $Z_j$ is a Pauli-Z operator acting on qubit $j$, the eigenvalue of $Z_j$ that corresponds to the eigenstate $\ket{x}$, i.e. $\expval{Z_j}{x}$,
    only depends on the $j$-th bit of $x$:
    \begin{equation}
        \expval{Z_j}{x} =
        \begin{dcases*}
            1 & if $x_j = 0$, \\
            -1 & if $x_j = 1$.
        \end{dcases*}
        \label{eqn:zi-cases}
    \end{equation}
    \label{lemma:eigval-zi}
\end{lemma}

\begin{proof}
    Given
    \begin{align}
        I & = \ketbra{0} + \ketbra{1} \\
        Z & = \ketbra{0} - \ketbra{1}.
    \end{align}
    The Pauli string notation $Z_j$ denotes a Pauli-Z operator at position $j$:
    \begin{align}
        Z_j & \equiv I_1\otimes I_2\otimes \cdots \otimes Z_j\otimes \cdots\otimes I_n \\
        \begin{split}
        & = \left(\ketbra{0}_1 + \ketbra{1}_1 \right) \otimes \left(\ketbra{0}_2 + \ketbra{1}_2 \right) \otimes \cdots \\
        &\quad \otimes \left(\ketbra{0}_j - \ketbra{1}_j\right) \otimes \cdots \otimes \left(\ketbra{0}_n + \ketbra{1}_n\right).
        \end{split}
    \end{align}
    The indices label the positions of the $I$ or $Z$ operators. After tensor product expansion, the term $\ketbra{1}_j$ in $Z_j$ will give a
    coefficient of $-1$ to the computational basis states that has $\ket{1}$ on position $j$, while the other terms will stay the same by having a 
    coefficient of $1$, hence proving the lemma.
\end{proof}

We then generalize to the weighted sum of one-local Pauli-Z using Lemma~\ref{lemma:eigval-zi}. Consider a 3-qubit system with 
$\tilde{H} = c_1Z_1 + c_2Z_2$, the eigenvalues of $\tilde{H}$ are listed in Table~\ref{tab:eig-htilde}.
Hence, we have the following lemmas.
\begin{table}
    \centering
    \caption{The eigenvalues of a weighted sum of one-local Pauli-Z $\tilde{H} = c_1Z_1+c_2Z_2$ in a 3-qubit system.}
    \begin{tabular}{cc}
        \toprule
        Eigenstate & Eigenvalue of $c_1Z_1+c_2Z_2$ \\
        \midrule
        $\ket{000}$ & $c_1+c_2$ \\
        $\ket{001}$ & $c_1+c_2$ \\
        $\ket{010}$ & $c_1-c_2$ \\
        $\ket{011}$ & $c_1-c_2$ \\
        $\ket{100}$ & $-c_1+c_2$ \\
        $\ket{101}$ & $-c_1+c_2$ \\
        $\ket{110}$ & $-c_1-c_2$\\
        $\ket{111}$ & $-c_1-c_2$\\
        \bottomrule
    \end{tabular}
    \label{tab:eig-htilde}
\end{table}

\begin{lemma}
    The eigenvalues of a weighted sum of one-local Pauli-Z [Eq.~(\ref{eqn:weighted-sum})]
    are the sum of the coefficient times $1$ or $-1$ depending on the corresponding bit:
    \begin{equation}
        \Expval*{\sum_{j=1}^n c_jZ_j}{x} = \sum_{j=1}^{n} (-1)^{x_j}c_j.
        \label{eqn:lemma2}
    \end{equation}
    \label{lemma:wsumz}
\end{lemma}

\begin{proof}
    The expectation follows linearity:
    \begin{equation}
        \Expval*{\sum_{j=1}^n c_jZ_j}{x} = \sum_{j=1}^n c_j \expval{Z_j}{x}.
    \end{equation}
    From Lemma~\ref{lemma:eigval-zi}, we can deduce
    \begin{equation}
        \expval{Z_j}{x} = (-1)^{x_j}.
    \end{equation}
    Hence,
    \begin{equation}
        \Expval*{\sum_{j=1}^n c_jZ_j}{x} = \sum_{j=1}^n (-1)^{x_j}c_j.
    \end{equation}
\end{proof}

\begin{lemma}
    If Pauli $Z_j$ is not in the weighted sum, i.e., $c_j=0$, then bit $j$ of the eigenstate $\ket{x}$ does not affect the eigenvalue of the weighted sum.
    \label{lemma:invariance}
\end{lemma}

\begin{proof}
    From Eq.~(\ref{eqn:lemma2}), if $c_j=0$, then $c_j$ does not contribute to the sum in the RHS of the equation.
    Observing the values in Table~\ref{tab:eig-htilde}, $c_j$ contributes different signs to the sum based on the value of $j$-th bit ($+1$ for $x_j=0$
    and $-1$ for $x_j=1$). If $c_j=0$, the sign has no effect on the sum ($\pm 0$ means the same), hence proving the lemma.
\end{proof}

Lemma~\ref{lemma:invariance} implies that only the measurements of the non-zero terms in $\tilde{H}$ are needed to obtain the eigenvalues.
Hence, the complexity to query the eigenvalues of $\tilde{H}$ is $\BigO(2^t)$, where $t$ is the number of non-zero Pauli-Z in $\tilde{H}$,
as there are at most $2^t$ distinct eigenvalues. In shot-based devices, measurements are taken only for the non-zero Pauli terms in $\tilde{H}$
to obtain probabilities for $2^t$ different eigenstates. In statevector simulation, the measurement result (a statevector)
is obtained with the size of $2^n$ ($n$ is the number of qubits), then the qubits that are not involved in $\tilde{H}$ can be traced out by taking
the partial trace on the statevector to reduce it to $2^t$ $(t< n)$. 
Combining Lemma~\ref{lemma:eigval-zi}, \ref{lemma:wsumz}, and \ref{lemma:invariance}, we arrive at the following:
\begin{theorem}
    The eigenvalue $\mu_x$ (corresponds to the eigenstate $\ket{x}$) of the $i$-th linear constraint Hamiltonian $H_i$ can be computed using
    \begin{align}
        \mu_x \equiv \expval{H_i}{x} & = C_i-\frac{1}{2}\expval{\tilde{H}_i}{x} \\
        & = C_i - \frac{1}{2}\sum_{j=1}^n (-1)^{x_j}a_{ij},
    \end{align}
    with $\BigO(2^t)$ complexity to query all the eigenvalues, where $t$ is the number of non-zero Pauli-Z in $\tilde{H}_i$.
    Then, the expectation of the penalty function, $\expval{\xi(H_i)}$, can be computed
    using Eq.~(\ref{eqn:exp-penalty}).
    \label{theorem:final}
\end{theorem}

\begin{proof}
    \begin{align}
        \mu_x & \equiv \expval{H_i}{x} \\
        & = \Expval*{C_i - \frac{1}{2}\sum_{j=1}^n a_{ij}Z_j}{x}\quad \text{[Eq.~(\ref{eqn:ci-hi})]}\\
        & = C_i - \frac{1}{2}\Expval*{\sum_{j=1}^n a_{ij}Z_j}{x} \\
        & = C_i - \frac{1}{2}\sum_{j=1}^n (-1)^{x_j}a_{ij}\quad \text{[Eq.~(\ref{eqn:lemma2})]}.
    \end{align}
    From Lemma~\ref{lemma:invariance}, we know that for a given qubit $j$, if $c_j=0$, then it does
    not contribute to the eigenvalue of the weighted sum.
    Therefore, there are only at most $2^t$ distinct eigenvalues in $\tilde{H}_i$, where $t$ is the number of non-zero Pauli-Z in $\tilde{H}_i$.
    Hence, the complexity to query all the eigenvalues in $\tilde{H}_i$ is $\BigO(2^t)$.
\end{proof}

\section{Experimental settings}
We use the Multiple Knapsack problem (MKP) as an example for our method.
Given $K$ knapsacks with limited capacities $W_i$ and $L$ items with respective values $v_j$ and weights $w_j$,
MKP seeks to assign items to the knapsacks such that the combination maximizes the value of the items carried,
such that the weights of the items $w_j$ do not exceed the capacity of the knapsack $W_i$, and each knapsack can contain only one item.
The binary decision variables are denoted as $x_{ij}\in\{0,1\}$, where $x_{ij}=1$ if item $j$ is placed in knapsack $i$ (and 0 otherwise).
It is formally defined as
\begin{align}
    \max_x & \sum_{i=1}^K \sum_{j=1}^L v_jx_{ij} \\
    \text{s.t.} & \sum_{j=1}^L w_jx_{ij} \leq W_i,\quad i = 1,...,K \\
    & \sum_{i=1}^K x_{ij}\leq 1,\quad j = 1,...,L.
\end{align}
where \\
\begin{tabular}{lll}
    $x_{ij}$& : & Decision variable to represent whether item $j$ is in \\
    & & knapsack $i$. \\
    $v_j$& : & Value for item $j$. \\
    $w_j$& : & Weight of item $j$. \\
    $W_i$& : & Capacity of knapsack $i$.
\end{tabular} \\

Note that any maximization problem can be converted to a minimization problem by taking the negative of the objective.
MKP is formulated into Eq.~(\ref{eqn:custom-ham}) and is solved using VQE. First, the constraints of MKP can be expressed as
\begin{align}
    h_i^{(1)}(x) & = \sum_{j=1}^L w_jx_{ij} - W_i \leq 0,\quad i = 1,...,m \label{eqn:mkp-const1} \\
    h_j^{(2)}(x) & = \sum_{i=1}^K x_{ij} - 1 \leq 0, \quad j = 1,...,n. \label{eqn:mkp-const2}
\end{align}
The problem can be formulated in the form of Eq.~(\ref{eqn:loss-custom}) as
\begin{equation}
    \mathcal{L}(x) = -\sum_{i=1}^K \sum_{j=1}^L v_jx_{ij} + \sum_{i=1}^{k} \lambda_i\ \xi[h_i(x)].
    \label{eqn:custom-mkp}
\end{equation}
Here, $h_i(x)$ includes both types of MKP constraints in Eq.~(\ref{eqn:mkp-const1}) and (\ref{eqn:mkp-const2}), and $k=K+L$.
Then, it is converted to the problem Hamiltonian in the form of Eq.~(\ref{eqn:custom-ham}) to be solved using VQE.

We solve 78 instances of MKP with at most 3 knapsacks and 4 items. After formulating to the custom loss function in Eq.~(\ref{eqn:custom-mkp}), we require 9 to 12 qubits to encode the problem as quantum states.
For VQE, we use a single layer hardware-efficient ansatz (HEA) to generate the trial state. We use a HEA that consists of parameterized rotational-Y $R_Y(\theta)$ gates
and unparameterized controlled-Z $CZ$ with linear entanglement ($CZ$ is placed between all adjacent qubits). This is a typical structure for HEA.
Note that we can no longer use a Hamiltonian-dependent ansatz under this setting, because there are terms with $\xi(H_i)$ and deciding the gates for the unitary $e^{-i\theta\xi(H_i)}$ might
be hard (unless it can be converted to Pauli sums, discussed in Sec.~\ref{sec:discussion}). The ansatz circuit is simulated using Qiskit statevector simulation~\cite{qiskit2024}. 
The optimization is done using the L-BFGS-B optimizer provided by SciPy~\cite{l-bfgs-b,scipy}, with \verb+maxfun+ (maximum number of function evaluations) and \verb+maxiter+ (maximum number of iterations) of 15000 each,
and \verb+ftol+ (function value tolerance) of $2.22\times 10^{-15}$.
The hyperparameters for the optimizer is set based on the default value provided in the Qiskit interface~\cite{bfgs-qiskit}.
Each instance of MKP has 3 random initializations (trials) of the parameters $\bm{\theta}$.

We use two types of penalty function for our experiment: the exponential function and the Heaviside step function.
For the exponential penalty, we formulate the loss function as
\begin{equation}
    \mathcal{L}_{\text{exp}}(x) = -\sum_{i=1}^K \sum_{j=1}^L v_jx_{ij} + \lambda_1 \sum_{i=1}^k e^{\lambda_2 h_i(x)}.
    \label{eqn:loss-exp}
\end{equation}
Instead of having individual penalty coefficients for each constraint, as such in Eq.~(\ref{eqn:custom-mkp}), we give only two penalty coefficients $\lambda_1$ and $\lambda_2$ for simplicity.
Similarly, the loss function for step penalty is 
\begin{equation}
    \mathcal{L}_\text{step}(x) = -\sum_{i=1}^K \sum_{j=1}^L v_jx_{ij} + \lambda\sum_{i=1}^k \Theta[h_i(x)],
    \label{eqn:loss-step}
\end{equation}
where $\Theta(\cdot)$ is the Heaviside step function:
\begin{equation}
    \Theta[h_i(x)] = 
    \begin{dcases*}
        1 & if $h_i(x) > 0$, \\
        0 & if $h_i(x) \leq 0$.
    \end{dcases*}
\end{equation}

It is important to show that there is a clear line of separation between feasible and infeasible solutions in Eq.~(\ref{eqn:loss-exp}) and (\ref{eqn:loss-step}), provided
that the penalty coefficient $\lambda$ is set correctly.
A typical method to use is the upper bound method~\cite{penalty-weights}, which finds the upper bound of the magnitude in the objective $f(x)$ by assigning $x_i=1$ for all $x_i$'s.
For MKP, it is the total value of all the items times the number of knapsacks $K$:
\begin{equation}
    \lambda_{\text{UB}} = K \sum_{j=1}^L v_j.
\end{equation}
Therefore, if the penalty $\lambda$ (or $\lambda_1$ in $\mathcal{L}_\text{exp}$) is set to $\lambda>\lambda_\text{UB}$, we can guarantee that the infeasible solutions will
have an objective at least as large as $\lambda_\text{UB}$ to offset them.

\section{Results}
We use three metrics to evaluate the effectiveness of our method: feasibility rate, optimality rate, and mean optimality gap.
The feasibility rate is the number of feasible solutions obtained out of the total number of instances solved.
The optimality rate is the number of optimal solutions obtained out of the total number of instances.
Feasible solution is a solution (sampled out of the optimal quantum circuit) that satisfies all constraints.
Optimal solution is a solution that is feasible and has the minimum (maximum) objective in the original problem.
Optimality gap is defined as the gap between the objective given by VQE and the optimal objective:
\begin{equation}
    \text{opt. gap} = 1 - \frac{C_\text{VQE}}{C_\text{opt}}.
\end{equation}
Lower optimality gap means better performance of VQE. 

\begin{table}
    \centering
    \caption{Feasiblity rate, optimality rate, and mean optimality gap for MKP modeled using the exponential penalty function with varied penalty coefficients $(\lambda_1,\lambda_2)$. %
    The red fonts show the highest value among different methods.}
    \begin{tabular}{cccc}
        \toprule
        Method / & Feasibility & Optimality & Mean \\
        Penalties $(\lambda_1,\lambda_2)$ & rate (\%) & rate (\%) & optimality gap\\
        \midrule
        Quadratic penalty~\cite{ub-penalty} & 77 & \bfred{40} & \bfred{0.067} \\
        \midrule
        (1, 3) & 27 & 21 & -0.287 \\
        (10, 3) & 80 & 11 & 0.218 \\
        (50, 3) & \bfred{99} & 0 & 0.995 \\
        \midrule
        (1, 5) & 48 & 22 & \bfred{0.067} \\
        (10, 5) & 69 & 13 & 0.228 \\
        (50, 5) & 82 & 1 & 0.564 \\
        \midrule
        (1, 10) & 46 & 5 & 0.224 \\
        (10, 10) & 46 & 2 & 0.266 \\
        (50, 10) & 46 & 2 & 0.274 \\
        \bottomrule
    \end{tabular}    
    \label{tab:exp-results}
\end{table}

Table~\ref{tab:exp-results} shows the results for the exponential penalty function. The performance under different penalty coefficients $(\lambda_1, \lambda_2)$ in Eq.~(\ref{eqn:loss-exp}) is evaluated.
A general trend shows that there is a trade-off between feasibility and optimality as $\lambda_1$ increases. As $\lambda_1$ increases, the feasibility rate increases, implying that the constraints are being enforced.
However, the decrease in optimality rate hints that the objective might have shifted by small values, as the constraints give different values for different values of the binary variables.
The feasibility rate achieves the highest of 99\% for the exponential function, whereas the quadratic penalty (unbalanced) still has the highest optimality rate of 40\%.
In terms of the mean optimality gap, the exponential penalty achieves 0.067, which is on par with the quadratic penalty. 
The negative optimality gap in $(1,3)$ arises from infeasible solutions but has a larger objective than the optimal objective.

\begin{figure}
    \centering
    \includegraphics[width=0.95\linewidth]{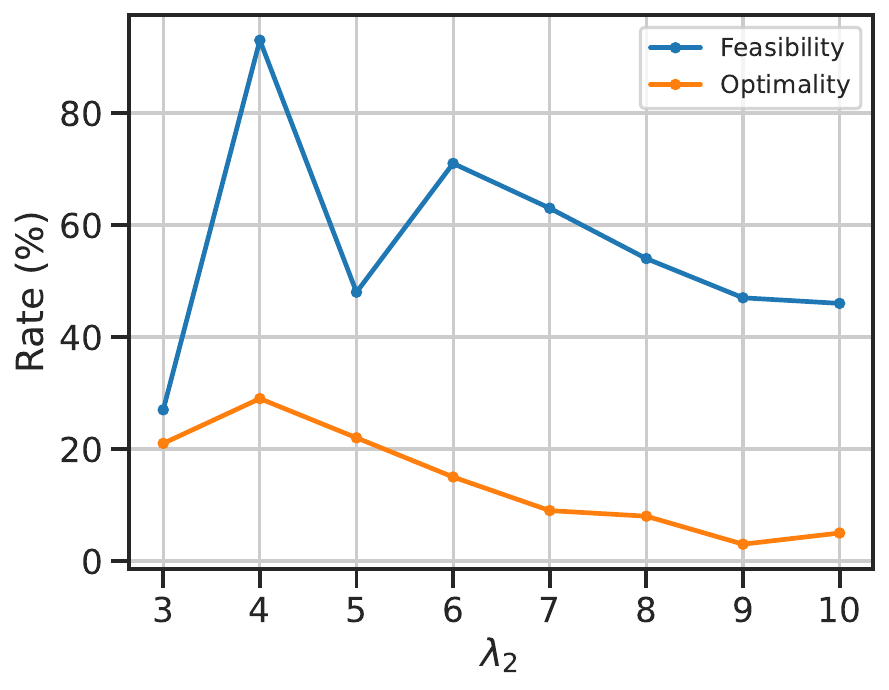}
    \caption{The variation of feasibility and optimality rate with the penalty coefficient $\lambda_2$ of the exponential penalty function. $\lambda_1$ is fixed to 1 in this case.}
    \label{fig:exp-lambda2}
\end{figure}

Fig.~\ref{fig:exp-lambda2} shows the change in the feasibility and optimality rate as the second penalty coefficient $\lambda_2$ is varied ($\lambda_1$ is fixed to 1).
As $\lambda_2$ increases, both the feasibility and optimality rate peak at $\lambda_2=4$, then decreases for larger values of $\lambda_2$.
This can be due to the bloated values of the penalty term, since $\lambda_2$ is in the exponent.

\begin{table}
    \centering
    \caption{Comparison of different methods: normal penalty with slack variables, unbalanced quadratic penalty, and our method with step penalty, on different metrics: feasibility rate, optimality rate, and %
    optimality gap. The red fonts show the highest value among different methods.}
    \begin{tabular}{cccc}
        \toprule
        Type of & Feasibility & Optimality & Mean \\
        penalty & rate (\%) & rate (\%) & optimality gap\\
        \midrule
        Normal (slack) & 63 & 29 & 0.139 \\
        Quadratic~\cite{ub-penalty} & 77 & 40 & \bfred{0.067} \\
        Step (ours) & \bfred{96} & \bfred{53} & 0.087 \\
        \bottomrule
    \end{tabular}
    \label{tab:step-results}
\end{table}

Since the results for the exponential penalty are not satisfying, we attempt to model the penalty function with a step function. 
The step function has a much clearer cut in the values of $h(x)$, as all infeasible solutions will add a penalty of $\lambda$ to the loss function, and feasible solutions will have zero contribution,
regardless of the values of $h(x)$. Table~\ref{tab:step-results} shows the results for MKP modeled with step penalty. We compare three types of penalty: normal formulation with slack variables in Eq.~(\ref{eqn:loss-slack}),
unbalanced quadratic penalty in Eq.~(\ref{eqn:loss-ub}), and the step penalty in Eq.~(\ref{eqn:loss-step}). The coefficient for the step penalty is set to $\lambda=50$.
It can be seen that our step penalty method shows the highest feasibility and optimality rate, despite having a slightly higher optimality gap than the quadratic penalty. 

\section{Discussion}\label{sec:discussion}
Although the step penalty shows better performance (in feasibility and optimality rate) than the quadratic unbalanced penalty and the normal penalty with slack variables,
its computational cost scales exponentially with the number of terms in a single constraint (Theorem~\ref{theorem:final}).
One possible way to work this out is to express the penalty function $\xi(H_i)$ in terms of Pauli sums.
For example, the stepped constraint Hamiltonian $\Theta(H_i)$ produces a matrix with eigenvalues of only 0's and 1's, and this can be decomposed into weighted sums of Pauli-Z.
Consider a 2-qubit example with
\begin{align}
    \Theta(H_i) & = 
    \begin{pmatrix}
        1 &&& \\
        & 0 && \\
        && 0 & \\
        &&& 0
    \end{pmatrix}\label{eqn:delta-ex} \\
    & = c_0I + c_1 Z_1 + c_2 Z_2 + c_3 Z_1Z_2. \label{eqn:weighted-sum-2qb}
\end{align}
By expanding the RHS of (\ref{eqn:weighted-sum-2qb}), the Pauli coefficients can be found by solving a system of linear equations: 
\begin{align}
    c_0 + c_1 + c_2 + c_3 &  = 1 \\
    c_0 - c_1 + c_2 - c_3 & = 0 \\
    c_0 + c_1 - c_2 - c_3 & = 0 \\
    c_0 - c_1 - c_2 + c_3 & = 0,
\end{align}
which is also 
\begin{equation}
    \begin{pmatrix}
        1 & 1 & 1 & 1 \\
        1 & -1 & 1 & -1 \\
        1 & 1 & -1 & -1 \\
        1 & -1 & -1 & 1
    \end{pmatrix}
    \begin{pmatrix}
        c_0 \\ c_1 \\ c_2 \\ c_3
    \end{pmatrix}
    = \begin{pmatrix}
        1 \\ 0 \\ 0 \\ 0
    \end{pmatrix}.
    \label{eqn:pauli-lse}
\end{equation}
Interestingly, the leftmost matrix is the Hadamard matrix on 2 qubits ($\mathcal{H}_2 = \mathcal{H}\otimes\mathcal{H}$) up to a global
factor of $1/2$, and the Hadamard matrix is self-inverse by a factor of $2^{-n}$:
\begin{equation}
    \mathcal{H}_n^{-1} = \frac{1}{2^n}\mathcal{H}_n,
\end{equation}
where $\mathcal{H}_n$ is the Hadamard matrix for $n$ qubits.
For a general $\Theta(H_i)$ of $n$ qubits, the Pauli coefficients $\mathbf{c}$ can be obtained by
\begin{equation}
    \mathbf{c} = \frac{1}{2^{n}} \mathcal{H}_n\ \vec{\Theta}(H_i),
\end{equation}
where $\vec\Theta(H_i)$ is the vectorized $\Theta(H_i)$.
Solving Eq.~(\ref{eqn:pauli-lse}), we can express $\Theta(H_i)$ in Eq.~(\ref{eqn:delta-ex}) as
\begin{equation}
    \Theta(H_i) = \frac{1}{4}(I + Z_1 + Z_2 + Z_1Z_2).
\end{equation}
However, doing this also requires the knowledge of all the eigenvalues of $\Theta(H_i)$. It would be ideal if the Pauli-Z expression of $H_i$ can be directly converted
to the Pauli-Z expression of $\Theta(H_i)$, which would appear challenging.

\section{Conclusion}
We proposed a method to formulate inequality-constrained problems into an unconstrained form with custom penalty functions without requiring the introduction of slack variables.
Two types of penalty functions are studied: the exponential function and the step penalty function. The exponential function returns the penalty of $e^{h(x)}$ (scaled by penalty coefficients),
which gives large values if the constraint is violated ($h(x)>0$) and gives small values if the constraint is satisfied ($h(x)\leq 0$). The step function returns 1 for solutions that violate the constraints
and return 0 for solutions that obey the constraints. We found that for exponential penalty, there is a trade-off between feasibility and optimality of solutions for different penalty coefficients.
As the penalty coefficient $\lambda_1$ increases, the feasibility tends to increase and the optimality tends to decrease. On the other hand, the step penalty function achieved higher feasibility and optimality rate 
compared to the slack-based penalty and the unbalanced quadratic penalty, while having a mean optimality gap that is close to that of the unbalanced penalty.

These improvements come at a cost that scales exponentially with respect to the number of terms in a constraint, but they eliminate the need for extra qubits, which is a critical advantage in the resource-limited
quantum hardware nowadays.
This would prove useful on problems that have far fewer variables involved in a single constraint compared to the number of variables in the entire problem.
For example, MKP requires $K\times L$ variables overall, but only at most $K$ or $L$ (whichever is larger) variables are in one constraint.
One promising future direction is to develop efficient decomposition techniques to express nonlinear penalties such as step functions in terms of Pauli operators, making them fully quantum-executable.
Overall, our results highlight the practicality and potential of implementing custom penalty strategies in gate-based quantum optimization and suggest a path toward more resource-efficient and accurate quantum solvers for
constrained combinatorial optimization problems.

\bibliography{ccop,qaoaref}
\bibliographystyle{ieeetr}
\end{document}